\documentclass[letterpaper,11pt]{article}

\usepackage[margin=1in]{geometry}  
\usepackage{authblk} 

\usepackage{amsmath,amssymb,amsthm}
\usepackage{tabularx}
\usepackage[algoruled,lined,algonl,procnumbered,nosemicolon]{algorithm2e}
\usepackage{lineno}
\usepackage{subcaption}  
\usepackage{float}
\usepackage{comment} 
\usepackage{enumitem} 
\usepackage[hyphens]{url} 
\usepackage{graphicx,hyperref} 

\usepackage{breakcites} 
\usepackage{tikz}
\usepackage{booktabs}
\usepackage[textsize=scriptsize]{todonotes}
\usepackage{lmodern}

\usepackage{orcidlink}
\usepackage{placeins}  
\usepackage{breakcites}

\colorlet{darkorange}{orange!85!black}
\colorlet{darkblue}{blue!75!black}
\colorlet{darkgreen}{green!50!black}

\newtheorem{theorem}{Theorem}
\newtheorem{corollary}{Corollary}

\newtheorem{lemma}{Lemma}
\newtheorem{observation}{Observation}
\newtheorem{remark}{Remark}
\newtheorem{conjecture}{Conjecture}
\newtheorem{definition}{Definition}

\newtheorem*{claim*}{Claim}

\newcommand{\PWS}{\textnormal{\textsc{Pinwheel Scheduling}}}
\newcommand{\PWPack}{\textnormal{\textsc{Pinwheel Packing}}}
\newcommand{\PWCov}{\textnormal{\textsc{Pinwheel Covering}}}
\newcommand{\VarPWCov}{\textnormal{\textsc{Variable Pinwheel Covering}}}

\newcommand{\tP}{\textnormal{3\textsc{-Partition}}}
\newcommand{\NTDM}{\textnormal{\textsc{Numerical}\ 3\textsc{-Dimensional Matching}}}

\newcommand{\EWPM}{\textnormal{\textsc{Exact Weighted Perfect Matching}}}
\newcommand{\EWPMshort}{\textnormal{\textsc{EWPM}}}

\newcommand{\EM}{\textnormal{\textsc{Exact Matching}}}

\newcommand{\RNTDM}{\textnormal{\textsc{Restricted Numerical}\ 3\textsc{-Dimensional Matching}}}
\newcommand{\RNTDMshort}{\textnormal{\textsc{RN3DM}}}

\newcommand{\kV}{\textnormal{\ensuremath{k}\textsc{-Visits}}}

\newcommand{\ILP}{\textnormal{\textsc{Integer Linear Programming}}}

\newcommand{\kVCov}{\textnormal{\ensuremath{k}\textsc{-Visits Covering}}}
\newcommand{\oneVCov}{\textnormal{1\textsc{-Visit Covering}}}
\newcommand{\twoVCov}{\textnormal{2\textsc{-Visits Covering}}}

\newcommand{\kVPack}{\textnormal{\ensuremath{k}\textsc{-Visits Packing}}}

\newcommand{\twoVPack}{\textnormal{2\textsc{-Visits Packing}}}
\newcommand{\threeVPack}{\textnormal{3\textsc{-Visits Packing}}}

\def\bigO{\ensuremath{\mathcal{O}}\xspace}

\makeatother

\begin{document}
\title{Finite Pinwheel Covering}

\author[1,2]{Sotiris Kanellopoulos \orcidlink{0009-0006-2999-0580} }

\affil[1]{National Technical University of Athens, Greece}
\affil[2]{Archimedes, Athena Research Center, Greece}

\affil[ ]{s.kanellopoulos@athenarc.gr}

\date{}

\maketitle

\begin{abstract}
In perpetual scheduling theory, the Pinwheel Covering problem asks, given~$n$ frequencies~$f_i$, whether there exists an infinite schedule such that every $f_i$ consecutive entries contain at most one occurrence of $i\in [n]$. This models $n$ agents taking turns at executing a job, with a recovery period before working again. Pinwheel Covering is, in a sense, the dual of Pinwheel Packing (also known as Pinwheel Scheduling), which similarly asks for \textit{at least} one occurrence of $i$ in every $f_i$ consecutive entries. The complexity of both problems is a major open question: both are known to be in PSPACE, but PSPACE-hardness remains unknown.

Recently, a finite version of Pinwheel Packing requiring only $k$ occurrences of $i\in [n]$ was introduced by [Kanellopoulos et al., SODA 2026~\cite{kVisits}] and proven to be strongly NP-complete. In this work we introduce $k$-Visits Covering, the analogous finite version of Pinwheel Covering, establishing strong NP-completeness even for $k=2$. As a corollary, we obtain that a generalization of Pinwheel Covering with varying frequencies is strongly NP-hard. To the best of our knowledge, this is the first strong NP-hardness result in the covering setting. We complement these results with a linear-time algorithm for $2$-Visits Covering with two distinct frequencies and a randomized polynomial-time algorithm when the number of distinct frequencies is constant. Lastly, we study the density thresholds of $k$-Visits Covering and prove that no non-trivial density bounds exist, contrasting the finite packing version.
\end{abstract}

\newpage

\section{Introduction}

\PWS\ is a fundamental problem in perpetual scheduling theory, asking whether $n$ tasks with frequencies~$f_i$ can be executed indefinitely, with each task~$i$ being executed at least once every $f_i$ time units. The problem was introduced by Holte, Mok, Rosier, Tulchinsky and Varvel~\cite{Holte_Pinwheel} in 1989, modeling periodic data retrieval from satellites, and has received significant attention since. In particular, a puzzling and long-standing open question regarding the largest value of \emph{density} (sum of inverse frequencies) that guarantees schedulability was answered by Kawamura~\cite{Kawamura_5/6_stoc} (STOC 2024), proving that this value is equal to $5/6$.

The \PWCov\ problem is, in a sense, dual to \PWS, asking for a schedule such that $i$ occurs \emph{at most} once every $f_i$ time units. A natural application for this variant involves $n$ agents taking turns at executing a job, with agent $i$ being able to work at most once every $f_i$ days, and the goal being to \emph{cover} every day with at least one agent. In this paper, we adopt the terminology of~\cite{Kawamura_pinwheel_cover}, referring to \PWS\ as \PWPack, in order to discern between the \emph{packing} and \emph{covering} settings.

The complexity of both \PWPack\ and \PWCov\ remains open. Both problems are known to be in PSPACE~\cite{Holte_Pinwheel,Pinwheel_ISAAC_2025}, but PSPACE-hardness remains open for both. In fact, even NP-hardness was open ever since the introduction of \PWPack\ in 1989, until Kleinberg and Mishra~\cite{kleinberg_mishra_NP_hardness} (FOCS 2026) very recently showed weak NP-hardness for both problems. A recent paper in SODA 2026 by Kanellopoulos, Pergaminelis, Kokkou, Markou and Pagourtzis~\cite{kVisits} defined \kVPack,\footnote{In~\cite{kVisits}, this problem is called \kV. In this work, we refer to it as \kVPack\ for clarity.} a version of \PWPack\ with finite amounts of task executions, and proved strong NP-completeness, contrasting the open status of the infinite version's complexity.

In this work, we define \kVCov, the analogous finite version of \PWCov, and prove strong NP-completeness. As a corollary, we obtain that \PWCov\ becomes strongly NP-hard if each frequency is allowed to change value after a given amount of occurrences in the schedule. We complement these results with some tractable special cases and density thresholds; see Section~\ref{sec:contributions} for an overview of our results.

The motivation for the study of finite pinwheel variants mainly stems from the open complexity of the respective infinite versions (note that strong NP-hardness and PSPACE-hardness remain open for both \PWPack\ and \PWCov, despite the recent breakthrough of Kleinberg and Mishra~\cite{kleinberg_mishra_NP_hardness}). Moreover, PSPACE-hardness proofs for periodic problems often involve modifying the NP-hardness proof of some finite version (cf.~\cite{Los_Alamos_periodic_PSPACE,PSPACE_dynamic,Papadimitriou_book}), although it is unclear if such a framework can be adapted to pinwheel problems. Even disregarding these open questions, finite pinwheel variants may be of interest due to their tractability in cases in which the infinite variants are intractable, as well as their ability to better capture realistic scenarios with finite repetitions.

\subsection{Related Work}

\PWS, also known as \PWPack\ (cf.~\cite{Kawamura_pinwheel_cover,Pinwheel_ISAAC_2025}), was first defined by Holte et al.~\cite{Holte_Pinwheel} in 1989. The problem has since received significant attention in the literature, with a long line of work~\cite{Bar-Noy_0.6,Chan_0.7,Fishburn_density,Gasieniec_towards_5/6} concerning the \emph{density threshold conjecture}~\cite{Chan_conjecture}, i.e., that all instances with sum of inverse frequencies bounded by $5/6$ are schedulable.\footnote{Note that the instance $\{2,3,x\}$ is unschedulable for all $x\in \mathbb{N}$, hence the $5/6$-threshold is tight.} This conjecture was confirmed by Kawamura~\cite{Kawamura_5/6_stoc} (STOC 2024) through a computer-assisted proof that involves verifying the schedulability of a finite set of instances. To this day, no proof exists for this density threshold without the use of a computer. The density threshold of $5/6$ has also been proven for a variant with real periods when there are up to three distinct frequencies, and is conjectured to hold in general~\cite{SOFSEM_real_periods}.
Another direction that has received attention involves the complexity of \PWPack, with Jacobs and Longo~\cite{Jacobs_Window_Scheduling_Complexity} and later Kobayashi and Lin~\cite{Pinwheel_ISAAC_2025} (ISAAC 2025) indicating that \PWPack\ is likely intractable, despite NP-hardness remaining open for over three decades. Recently, Kleinberg and Mishra~\cite{kleinberg_mishra_NP_hardness} finally showed weak NP-hardness. Strong NP-hardness and PSPACE-completeness remain open, although a generalization with task durations~\cite{Feinberg_Generalized_Pinwheel,SOFSEM_Kusano_pinwheel_durations} and a generalization in edge-weighted graphs~\cite{PSPACE_UAV} are known to be strongly NP-hard and PSPACE-complete respectively.

An optimization version of \PWPack\ known as \emph{Bamboo Garden Trimming} (BGT) has also been extensively studied. A long line of work~\cite{Bamboo_second,Bamboo_approx_2,Bamboo_first,Bamboo_approx_1,Kawamura_5/6_stoc} shows various constant approximations for this problem, with Mishra~\cite{Patrolling_SODA_2026} (SODA 2026) establishing a $9/7$-approximation. A combinatorial variant of BGT has been studied by Mendoza-Cadena, Merino, Nielsen and Schewior~\cite{Schewior_Combinatorial_Perpetual_Scheduling_ICALP} (ICALP 2026). 

\PWCov\ is less studied than \PWPack, although some results and open questions are shared between the two (cf.~\cite{Pinwheel_ISAAC_2025}); notably, strong NP-hardness and PSPACE-hardness are open for both problems, while membership in PSPACE is known for both. \PWCov\ was first introduced by Kawamura and Soejima~\cite{Kawamura_Soejima_point_patrolling} as \emph{Point Patrolling}. Kawamura, Kobayashi and Kusano~\cite{Kawamura_pinwheel_cover} (CIAC 2025) studied density thresholds for \PWCov\ and proved NP-hardness for a variant with exact frequencies. Mishra~\cite{Patrolling_SODA_2026} (SODA 2026) and Kawamura and Kobayashi~\cite{Kawamura_Kobayashi_density_covering} (ESA 2026) concurrently proved a tight density threshold of $\sum_{i=0}^\infty 1/(2^i+1)\approx 1.264$ for \PWCov, i.e., that every \PWCov\ instance with density no less than this value is schedulable.

Kanellopoulos et al.~\cite{kVisits} (SODA 2026) defined \kVPack\ as a finite version of \PWPack\ and showed that (i)~\twoVPack\ is strongly NP-complete, (ii)~\PWPack\ becomes strongly NP-hard if frequencies are allowed to change after a given amount of visits to the respective task, and (iii)~\twoVPack\ is tractable when all frequencies are distinct, or when there are up to two distinct frequencies, or when the input is sparse in a sense. Kanellopoulos, Mitropoulos, Pergaminelis and Tolias~\cite{ICALP_kVisits} (ICALP 2026) further showed that (i)~\twoVPack\ remains strongly NP-complete for maximum multiplicity $2$ (despite being in P for maximum multiplicity $1$), (ii)~\twoVPack\ admits a randomized polynomial-time algorithm when the number of distinct frequencies is constant, and (iii)~\twoVPack\ is always schedulable for inputs with density at most $\sqrt{2}-1/2\approx 0.9142$. Note that this density threshold is larger than the $5/6\approx 0.8333$ threshold proven for the infinite version~\cite{Kawamura_5/6_stoc}, although it does not seem to be tight. As $k \to \infty$, this density threshold of \kVPack\ approaches $5/6$~\cite{ICALP_kVisits}.

\subsection{Our Contributions}\label{sec:contributions}

In this work, we define \kVCov\ as a version of \PWCov\ with finite repetitions, analogous to the finite packing variant of~\cite{kVisits}. Formally, the problem asks to cover $kn$ time slots with $k$ occurrences (visits) of each of the numbers in $\{1,\ldots,n\}$, such that consecutive visits of $i\in \{1,\ldots,n\}$ are at least $f_i$ time slots apart.

Our main result is that \twoVCov\ is strongly NP-complete, through a reduction from \RNTDM\ (\RNTDMshort)~\cite{Flow_shop_Yu}. As a corollary from this, we obtain that \PWCov\ becomes \emph{strongly} NP-hard when frequencies are allowed to change their values after a given amount of visits. To the best of our knowledge, this is the first strong NP-hardness result for a \PWCov\ variant.

A crucial component for our hardness proof for \twoVCov\ involves a \emph{disconnection} property (Lemma~\ref{lem:disconnection_property}), roughly stating that the positions of first visits and the positions of second visits are distinct for all instances. This property is analogous to the one of \twoVPack~\cite{kVisits}, albeit simpler. Such a property essentially transforms the problem into a numerical matching variant with inequalities, rendering reductions from and to numerical matching problems possible.

We complement our hardness results with a linear-time algorithm for \twoVCov\ with up to two distinct frequencies, as well as a proof that \twoVCov\ is in the complexity class~$\mathrm{RP}$ (randomized polynomial-time) when the number of distinct frequencies is constant. We achieve the latter through a reduction to \EWPM\ (see Def.~\ref{def:EWPM}).

Lastly, we prove that no non-trivial density threshold exists for \kVCov, contrasting the results of Kanellopoulos et al.~\cite{ICALP_kVisits} (ICALP 2026) for \kVPack. Specifically, we prove (for all $k\geq 2$) that all \kVCov\ instances with density less than~$1$ are unschedulable, and that unschedulable instances exist even for arbitrarily large density. Note that the former is not immediate for the finite version (despite being trivial for \PWCov~\cite{Kawamura_pinwheel_cover}); we prove it by carefully applying the arithmetic-harmonic mean inequality to obtain a density bound.

\section{Preliminaries}

\subsection{Common notation}

Throughout the paper, we use the notations $[n]=\{1,\ldots,n\}$ and $[m,n]=\{m,\ldots,n\}$, for $m,n\in \mathbb{N}$ with $m\leq n$. For a (multi)set $A$ and $c \in \mathbb{N}$, we use the notations $A + c = \{a+c\mid a\in A\}$ and $A - c = \{a-c\mid a\in A\}$.\footnote{If $A$ is a multiset, we assume this operation preserves the cardinality of each element, e.g., $\{3,4,4\}+2=\{5,6,6\}$.}
If $c\geq \max(A)$ we may also use the notation $c-A= \{c-a\mid a\in A\}$. Additionally, for a (multi)set $A$, we use the notation $\sum(A)=\sum_{a\in A}a$. Throughout the paper, whenever a (multi)set is given as input to a computational problem, we assume that it is sorted in non-decreasing order; thus, any $\bigO( n\log n)$ time factors that would be added to algorithms for sorting are omitted.

\subsection{Problem definitions}

First, let us define \PWCov, which was originally introduced as \textit{point patrolling} in~\cite{Kawamura_Soejima_point_patrolling}, as well as the respective version with finite occurrences (visits) that we will study in this work.

\begin{definition}[\PWCov]\label{def:PWCov}
    Given a (multi)set of positive integers (frequencies) $F=\{f_1,\ldots,f_n\}$, the \PWCov\ problem asks whether there exists an infinite schedule visiting one $i\in [n]$ per time unit, such that every $f_i$ consecutive entries contain at most one visit of~$i$, for all $i\in [n]$.
\end{definition}

\begin{definition}[\kVCov]\label{def:kVCov}
    Given a (multi)set of positive integers (frequencies) $F=\{f_1,\ldots,f_n\}$, the \kVCov\ problem asks whether there exists a schedule of length $nk$ visiting one $i\in [n]$ per time unit, such that each $i\in [n]$ is visited exactly $k$ times and every $f_i$ consecutive entries contain at most one visit of $i$, for all $i\in [n]$.
\end{definition}

\kVCov\ is defined in a similar manner to the \kVPack\ problem defined in~\cite{kVisits} as a finite version of \PWPack.
It is immediate from Definition~\ref{def:kVCov} that the answer to the \oneVCov\ problem is always yes: since the schedule contains exactly one occurrence of each $i\in [n]$, the desired condition for the consecutive entries is always satisfied. As such, this work will primarily focus on \twoVCov.

\section{Computational complexity of \kVCov}

In this section we prove that \twoVCov\ is strongly NP-complete via a reduction from \RNTDM\ (\RNTDMshort)~\cite{Flow_shop_Yu}. As a corollary, we transfer strong NP-hardness to a generalization of \PWCov\ where each frequency $f_i$ varies depending on the amount of times the respective number $i\in [n]$ has occurred in the schedule.

\subsection{Disconnection property}

We start by proving a property that disconnects the schedule positions in which first and second visits have to be placed in \twoVCov, inspired by the disconnection property proved for \twoVPack\ in~\cite{kVisits}. This essentially transforms \twoVCov\ into a numerical matching variant, paving the way for proving NP-hardness.

\begin{lemma}[Disconnection property]\label{lem:disconnection_property}
    A \twoVCov\ instance $F=\{f_1,\ldots,f_n\}$ admits a schedule if and only if it admits a schedule such that the first visits of all $i\in [n]$ are placed in some permutation of the positions $1,\ldots,n$. (Equivalently: the second visits of all $i\in [n]$ are placed in some permutation of the positions $n+1,\ldots,2n$.)
\end{lemma}

\begin{proof}
    The converse direction is trivial. For the forward direction, assume that $F$ admits a schedule~$S$ violating the desired property, i.e., there exists some $i\in [n]$ whose first visit is placed in position $p>n$ in~$S$. This implies that there exists some $j \in [n]$, $j\neq i$, whose second visit is placed in position $q\leq n$ in~$S$. We construct a schedule $S'$ by swapping the contents of positions $p,q$ in~$S$. This causes the first visit of~$i$ to be placed earlier and the second visit of~$j$ to be placed later; hence, it increases both the distance between the two visits of~$i$ and the distance between the two visits of~$j$, without affecting the placements of any numbers other than~$i,j$. Since~$S$ is a feasible schedule for instance~$F$, we infer that~$S'$ is also a feasible schedule for~$F$.

    By repeatedly applying the described transformation\footnote{Notice that the described transformation does not affect the placements of numbers other than the violators $i,j$. Thus, applying this transformation at most $n$ times causes the desired property to be satisfied.} as long as there exists some $i\in [n]$ violating the desired property, we can obtain a schedule in which the first visits of all $i\in [n]$ are placed in (some permutation of) the positions $1,\ldots,n$, proving the lemma.
\end{proof}

With Lemma~\ref{lem:disconnection_property} we can redefine \twoVCov\ as a numerical matching variant as follows.

\begin{definition}[\twoVCov]\label{def:twoVCov}
    Given a (multi)set of positive integers (frequencies) $F=\{f_1,\ldots,f_n\}$, the \twoVCov\ problem asks whether there exists a subset $M$ of $F\times [n]\times [n+1,\ 2n]$ such that:
    \begin{itemize}
        \item Every $f_i \in F$, $b \in [n]$, $c \in [n+1,\ 2n]$ occurs exactly once in~$M$.
        \item For every triplet $(f,b,c)\in M$, it holds that $f+b \leq c$.
    \end{itemize}
\end{definition}

Intuitively, Definition~\ref{def:twoVCov} states that each frequency has to be matched with a position in $[n]$ for its first visit and a position in $[n+1,\ 2n]$ for its second visit, such that the distance between the two visits is at least the frequency. Observe that, by Lemma~\ref{lem:disconnection_property}, Definition~\ref{def:twoVCov} is equivalent to Definition~\ref{def:kVCov} for $k=2$.

\subsection{Reducing \RNTDMshort\ to \twoVCov}

We will now reduce \RNTDM\ (\RNTDMshort) to \twoVCov\ in order to prove strong NP-completeness for the latter. First, we formally define the \RNTDMshort\ problem.\footnote{This problem is called ``restricted'' because two out of three input sets are fixed to $[n]$, instead of being arbitrary as in regular \NTDM.}

\begin{definition}[\RNTDMshort~\cite{Flow_shop_Yu}]\label{def:RN3DM}
    Given a (multi)set of positive integers $A=\{a_1,\ldots,a_n\}$ and an integer $\sigma$ such that \[\sum(A) + n(n+1) = n\sigma,\] the \RNTDM\ (\RNTDMshort) problem asks whether there exists a subset $M$ of $A\times [n]\times [n]$ such that:
    \begin{itemize}
        \item Every $a_i \in A$, $b \in [n]$, $c \in [n]$ occurs exactly once in~$M$.
        \item For every triplet $(a,b,c)\in M$, it holds that $a+b+c = \sigma$.
    \end{itemize}
\end{definition}

Note that $\sigma$ can be inferred from $A$ and can thus be omitted from the problem's input.

\begin{theorem}[Yu, Hoogeveen, Lenstra 2004~\cite{Flow_shop_Yu}]\label{thrm:RN3DM_hardness}
    \RNTDMshort\ is strongly NP-complete.
\end{theorem}

Figures~\ref{fig:RN3DM} and~\ref{fig:2Vcov} show a sketch of our reduction (Theorem~\ref{thrm:2Vcov_hardness}) with an example. Note that, since \RNTDMshort\ is \emph{strongly} NP-complete, $\sigma$ can be assumed to be polynomial in $n$ for our reduction; this is crucial for proving that the reduction runs in polynomial time, since the instance is padded with an amount of elements which is dependent on $\sigma$.

\begin{figure}[ht]
     \centering
     \begin{subfigure}{0.46\textwidth}
         \centering
         \includegraphics[width=\linewidth, page=1]{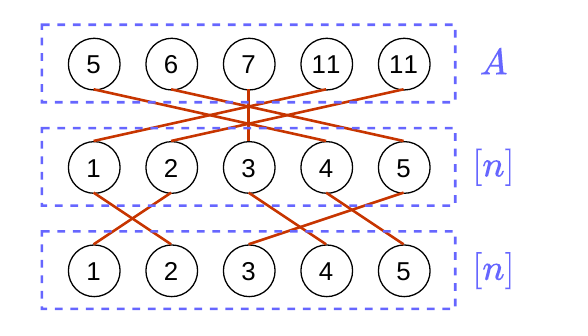}
         \caption{\RNTDMshort\ instance and solution.}
         \label{fig:RN3DM_1}
     \end{subfigure}\hfill
     \begin{subfigure}{0.54\textwidth}
         \centering
         \includegraphics[width=\linewidth, page=2]{RN3DM_covering.drawio.pdf}
         \caption{Equivalent formulation with target sums.}
         \label{fig:RN3DM_2}
     \end{subfigure}
     \caption{The \RNTDMshort\ instance $A=\{5,6,7,11,11\}$ and its solution (a), matching all numbers in triplets with sum $\sigma = 14$. In (b), we show an equivalent formulation with the third set being used as target sums, e.g., $11+2=13$ (cf.~\cite{temporal_path_cover,telephone_broadcasting_NMTS} for reductions involving such numerical matching variants). Observe that the solution is preserved (red lines).}
     \label{fig:RN3DM}
\end{figure}

\begin{figure}[ht]
    \centering
    \includegraphics[width=0.85\linewidth, page=3]{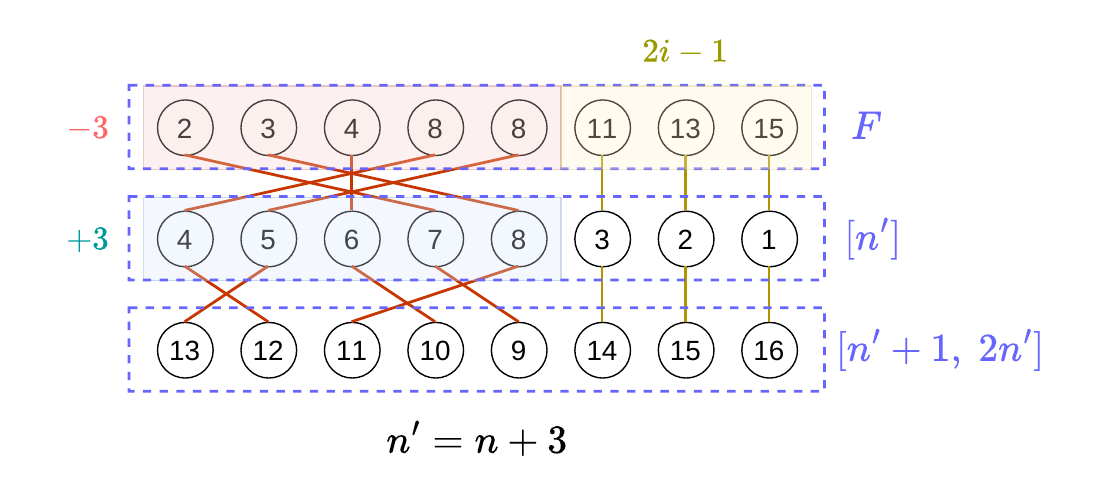}
    \caption{An example of our reduction in Theorem~\ref{thrm:2Vcov_hardness}, from the \RNTDMshort\ instance $A=\{5,6,7,11,11\}$ of Figure~\ref{fig:RN3DM} with $n=5$, $\sigma = 14$, to the \twoVCov\ instance $F$ of size~$n'$. The numbers in~$A$ are decreased by $d=\sigma-2n-1=3$, while the numbers in the second row are increased by~$3$ in order to preserve the solution (red lines).  Recall that by Def.~\ref{def:twoVCov}, the second and third row must contain the numbers in $[n']$ and $[n'+1,\ 2n']$ respectively. Hence, we pad the set of frequencies with numbers of the form $2i-1$, forcing a compulsory matching with the unmatched numbers of $[n']$ and $[n'+1,\ 2n]$ in the second and third row (vertical gold lines). A crucial component in the correctness of the reduction is that, when the sum of the first two rows matches the sum of the third row, a perfect matching in triplets $(a,b,c)$ with $a+b\leq c$ is equivalent to a perfect matching in triplets with $a+b = c$.}
    \label{fig:2Vcov}
\end{figure}

The following auxiliary lemmas will be useful for our reduction.

\begin{lemma}\label{lem:RN3DM_addition}
    Let $A$ be a (multi)set of positive integers and $x\in \mathbb{Z}$ with $x \geq -\min(A)+1$.\footnote{The purpose of this inequality is to ensure that $A+x$ consists of positive integers.}
    Then, the \RNTDMshort\ instances $A$ and $A+x$ are equivalent.
\end{lemma}

\begin{proof}
    By definition, $A$ admits a solution iff there exists a (perfect) matching $M\subseteq A \times [n] \times [n]$ such that all triplets in $M$ have equal sums. It is clear that by adding $x$ to all elements in $A$ the existence of such a solution is preserved, and vice versa by subtracting $x$ from all elements in~$A+x$.
\end{proof}

\begin{lemma}\label{lem:RN3DM_minimum_element}
    For every non-trivial \RNTDMshort\ instance it holds that $\min(A) \geq \sigma - 2n$.
\end{lemma}

\begin{proof}
    It suffices to observe that every \RNTDMshort\ instance with $\min(A) + 2n < \sigma$ is a trivial no-instance, since there exist no numbers $(b,c)\in [n]\times [n]$ such that $\min(A)+b+c = \sigma$.
\end{proof}

\begin{lemma}\label{lem:RN3DM_sigma}
    \RNTDMshort\ is strongly NP-complete even when restricted to instances such that $\sigma \geq 2n +1$.
\end{lemma}

\begin{proof}
    Let $A=\{a_1,\ldots,a_n\}$.
    By Lemma~\ref{lem:RN3DM_addition}, $A$ and $A'=A+n$ are equivalent \RNTDMshort\ instances. For~$A'$ we have $\sigma' = \left(\sum(A')/n\right) +n+1\geq 2n+1$. Hence, every \RNTDMshort\ instance can be transformed in polynomial time into an equivalent one that satisfies the desired property, while preserving polynomial values. In conjunction with Theorem~\ref{thrm:RN3DM_hardness}, this concludes the proof.
\end{proof}

We are now ready to present our main reduction.

\begin{theorem}[Reduction]\label{thrm:2Vcov_hardness}
    \twoVCov\ is strongly NP-complete.
\end{theorem}

\begin{proof}
    Since membership in NP is trivial for \twoVCov, we focus on proving strong NP-hardness by reducing \RNTDMshort\ to \twoVCov. Due to Lemma~\ref{lem:RN3DM_sigma}, we can assume an \RNTDMshort\ instance $A=\{a_1,\ldots,a_n\}$, $\sum(A) + n(n+1) =n\sigma$, satisfying
    \begin{equation}\label{eq:sigma}
        \sigma \geq 2n +1.
    \end{equation}
    Additionally, by Lemma~\ref{lem:RN3DM_minimum_element} we have
    \begin{equation}\label{eq:min_element}
        \min(A) \geq \sigma -2n.
    \end{equation}
    
    Based on $A$, we construct an instance $F=\{f_1,\ldots,f_{n'}\}$ of \twoVCov. Define $d=\sigma - 2n -1$ and observe that $d\geq 0$ by~\eqref{eq:sigma}. We set $n'=n+d$ and define $F$ as follows:

    \[f_i = \begin{cases}
        f_i=a_i - d, & i\in [n] \\
        f_i=2i-1, & i\in [n+1,\ n']
    \end{cases}.\]
    Note that $f_i>0$, $\forall i\in [n']$, since by~\eqref{eq:min_element} we have
    \begin{equation*}
        \min(A) \geq \sigma - 2n \iff \min(A) - d >0.
    \end{equation*}
    Thus, $F$ is a well-defined \twoVCov\ instance. This concludes the construction.

    We now prove the correctness of the reduction, i.e., that $A$ has a solution if and only if $F$ admits a \twoVCov\ schedule.

    ($\Rightarrow$) Suppose $A$ has a solution, i.e., there exists a (perfect) matching $M\subseteq A\times [n] \times [n]$ with $a+b+c = \sigma$, $\forall (a,b,c)\in M$. We can equivalently restate this as follows:\footnote{In all subsequent relations, whenever we write $\exists M\subseteq \ldots$, we imply that $M$ is a perfect $3$-dimensional matching between the three respective sets, as per Def.~\ref{def:twoVCov} and~\ref{def:RN3DM}. For the rest of this proof, we omit this condition for the sake of readability.}
    \begin{align}
        & \exists M\subseteq A \times [n] \times [n]: a+b+c = \sigma,\ \forall (a,b,c)\in M. \notag \\
        \iff & \exists M\subseteq A \times [n] \times (\sigma - [n]): a+b=c,\ \forall (a,b,c)\in M. \notag \\
        \iff & \exists M\subseteq (A-d) \times ([n]+d) \times (\sigma - [n]): a+b=c,\ \forall (a,b,c)\in M. \notag \\
        \iff & \exists M\subseteq (A-d) \times [d+1,\ n'] \times [\sigma - n,\  \sigma - 1]: a+b=c,\ \forall (a,b,c)\in M. \label{eq:matching_original}
    \end{align}
    By Lemma~\ref{lem:disconnection_property}, we can use Definition~\ref{def:twoVCov} for \twoVCov. Hence, it suffices to prove that there exists a (perfect) matching $M^*\subseteq F \times [n'] \times [n'+1,\ 2n']$ such that $a+b\leq c,\ \forall (a,b,c)\in M^*$. Observe that:
    \begin{itemize}
        \item $A-d \subseteq F$ (by construction).
        \item $[d+1,\ n'] \subseteq [n']$.
        \item $[\sigma - n,\  \sigma - 1] \subseteq [n'+1,\ 2n']$, since $2n' = (\sigma-1) + (\sigma-2n-1)\geq \sigma - 1$ by~\eqref{eq:sigma}.
    \end{itemize}
    From this, it follows that the matching of~\eqref{eq:matching_original} can be used as part of $M^*$.\footnote{Notice that the $a+b=c$ condition of~\eqref{eq:matching_original} is stricter than the desired $a+b \leq c$ condition of $M^*$.} It now suffices to prove that the remaining elements can also be matched between them, i.e., that
    \begin{align}
        & \exists M\subseteq (F\setminus (A-d)) \times ([n'] \setminus [d+1,\ n'])  \times ([n'+1,\ 2n']\setminus [\sigma - n,\  \sigma - 1]): a+b\leq c,\ \forall (a,b,c)\in M. \notag \\
        & \iff \exists M\subseteq \bigcup_{i\in[n+1,\ n']}\{2i-1\} \times [d] \times [\sigma,\ 2n']: a+b\leq c,\ \forall (a,b,c)\in M. \label{eq:matching_remaining}
    \end{align}
    To prove~\eqref{eq:matching_remaining}, it suffices to match the elements of the sets $S_1=\bigcup_{i\in[n+1,\ n']}\{2i-1\}$, $S_2=[d]$ and $S_3=[\sigma,\ 2n']$ as follows.
    \begin{itemize}
        \item Match $\max(S_1)$ with $\min(S_2)$ and $\max(S_3)$ (i.e., $(2n'-1) + 1 = 2n'$).
        \item Remove the matched elements and repeat the above. In each step, $\max(S_1)$ decreases by $2$, $\min(S_2)$ increases by $1$ and $\max(S_3)$ decreases by $1$, thus preserving the $\max(S_1)+\min(S_2)=\max(S_3)$ condition.
    \end{itemize}
    Thus, we have proven that $M^*$ exists, concluding the forward direction of the correctness proof.

    ($\Leftarrow$) Suppose that $F$ admits a \twoVCov\ schedule; by Lemma~\ref{lem:disconnection_property}, we have:
    \begin{equation}\label{eq:matching_schedule}
        \exists M^*\subseteq F \times [n'] \times [n'+1,\ 2n']: a+b\leq c,\ \forall (a,b,c)\in M^*.
    \end{equation}

    First, we prove that, in $M^*$, the elements in $S=F\setminus (A-d)=\bigcup_{i\in[n+1,\ n']}\{2i-1\}$ must be matched with the $d$ smallest elements of $B=[n']$ and the $d$ largest elements of $C=[n'+1,\ 2n']$. Intuitively, this means that the vertical gold lines in Figure~\ref{fig:2Vcov} are compulsory.
    \begin{itemize}
        \item When $\max(S)=2n'-1$ is matched with any element in $B$, it creates a sum at least as large as $\max(C)=2n'$. Hence, $\max(S)$ and $\max(C)$ must be included in the same triplet. Since $\max(C)-\max(S)=\min(B)=1$, we obtain that it must be $(\max(S),\min(B),\max(C))\in M^*$.
        \item If we remove the above matched elements from $S$, $B$ and $C$, then $\max(S)$ decreases by~$2$, $\min(B)$ increases by $1$ and $\max(C)$ decreases by $1$. Then, the same holds for the new elements: $\max(S)$ matched with any available element in $B$ creates a sum at least as large as $\max(C)$. Thus, $\max(S)$ must be matched with $\min(B)$ and $\max(C)$ in $M^*$. We can repeat this procedure until all elements in $S$ are matched with the~$d$ smallest elements of $B$ and the~$d$ largest elements of $C$.
    \end{itemize}

    From the above we infer that~\eqref{eq:matching_schedule} is equivalent to:
    \begin{equation}\label{eq:matching_schedule_reduced}
        \exists M\subseteq (A-d) \times [d+1,\ n'] \times [\sigma-n,\ \sigma-1]: a+b\leq c,\ \forall (a,b,c)\in M.
    \end{equation}

    However, by definition of the \RNTDMshort\ problem it holds that
    \begin{align}\label{eq:tight_sums}
         \sum(A)+n(n+1) & =n\sigma. \notag \\
        \iff \sum(A-d)+\sum([d+1,\ n']) & =\sum([\sigma-n,\ \sigma-1]).
    \end{align}

    Thus, for all $n$ triplets $(a,b,c)\in M$ in~\eqref{eq:matching_schedule_reduced} to satisfy $a+b \leq c$, it is mandatory for all of them to satisfy $a+b = c$, since the sums of the respective sets are, in a sense, ``tight'' by~\eqref{eq:tight_sums}. We obtain that
    \begin{equation*}
        \exists M\subseteq (A-d) \times [d+1,\ n'] \times [\sigma-n,\ \sigma-1]: a+b= c,\ \forall (a,b,c)\in M.
    \end{equation*}
    By increasing all elements in the first set by $d$ and decreasing all elements in the second set by $d$, we can equivalently restate this as
    \begin{equation*}
        \exists M\subseteq A \times [n] \times [\sigma-n,\ \sigma-1]: a+b= c,\ \forall (a,b,c)\in M,
    \end{equation*}
    which is trivially equivalent to
    \begin{equation*}
        \exists M\subseteq A \times [n] \times [n]: a+b+c=\sigma,\ \forall (a,b,c)\in M.
    \end{equation*}
    Thus, we have proven that the \RNTDMshort\ instance $A$ admits a solution.

    This concludes the correctness of the reduction. Since \RNTDMshort\ is strongly NP-complete (Theorem~\ref{thrm:RN3DM_hardness}, Lemma~\ref{lem:RN3DM_sigma}), we can assume that $a_1,\ldots,a_n$ are polynomial in $n$. This implies that $\sigma$ is also polynomial in $n$ and, thus, the reduction runs in polynomial time and $F$ consists of values polynomial in $n$.
    We conclude that the described reduction preserves strong NP-hardness.
\end{proof}

\subsection{A corollary for the infinite version}

We will now use Theorem~\ref{thrm:2Vcov_hardness} to prove that \PWCov\ becomes strongly NP-hard if we allow frequencies~$f_i$ to change after a given amount of visits to $i$. First, we present a formal definition for this generalization of \PWCov.

\begin{definition}[\VarPWCov]\label{def:VarPWCov}
    Given an integer $t>0$ and two (multi)sets of positive integers (frequencies) $F=\{f_1,\ldots,f_n\}$, $F'=\{f_1',\ldots,f_n'\}$, the \VarPWCov\ problem asks whether there exists an infinite schedule visiting one $i\in [n]$ per time unit, such that for all $i\in [n]$, if $p_i$ is the position containing the $t$-th visit of $i$, it holds that:
    \begin{itemize}
        \item In positions $1,\ldots,p_i$, every $f_i$ consecutive entries contain at most one visit of~$i$.
        \item From position $p_i$ onward, every $f_i'$ consecutive entries contain at most one visit of~$i$.
    \end{itemize}
\end{definition}

\begin{theorem}\label{thrm:VAR_PWC}
    \VarPWCov\ is strongly NP-hard.
\end{theorem}

\begin{proof}
    We will reduce \twoVCov\ to \VarPWCov. Let $A=\{a_1,\ldots,a_n\}$ be a \twoVCov\ instance. We construct a \VarPWCov\ instance $V=\left(F=\{f_1,\ldots,f_n,f_{n+1}\},F'=\{f_1',\ldots,f_n',f_{n+1}'\},t\right)$ as follows.
    \begin{itemize}
        \item $f_i=a_i$ and $f_i'=3n$ for $i\in [n]$.
        \item $f_{n+1}=2n+1$ and $f_{n+1}'=1$.
        \item $t=2$.
    \end{itemize}
    We now prove that $A$ admits a schedule if and only if $V$ admits a schedule.

    ($\Rightarrow$) Suppose $A$ admits a schedule $S$ of length $2n$, visiting each $i\in [n]$ exactly twice. We construct a schedule for $V$ as follows, making use of the fact that $n+1$ can cover all positions after its second visit, since $t=2$ and $f_{n+1}'=1$.
    \begin{itemize}
        \item Place the first visit of $n+1$ in position $1$ and its second visit in position $2n+2$.
        \item For positions $2$ up to $2n+1$, copy the schedule $S$.
        \item Fill all positions after $2n+2$ with $n+1$.
    \end{itemize}
    It follows immediately that, if $S$ is feasible for $A$, then the proposed schedule is feasible for $V$.

    ($\Leftarrow$) Suppose $V$ admits an infinite schedule $S$. We will prove that $A$ admits a schedule $S'$ of length $2n$.
    First, observe that, since $t=2$ and $f_{n+1}=2n+1$, it follows that $n+1$ can cover at most two positions in $[2n+2]$ in $S$. Thus, the remaining $2n$ positions in $[2n+2]$ have to be covered by $1,\ldots,n$. 
    However, since $t=2$ and $f_i'=3n$ for $i\in [n]$, it follows that each $i\in [n]$ can cover at most two positions in $[2n+2]$ in $S$. Thus, we have proven that every $i \in [n+1]$ can cover at most two positions of the positions $[2n+2]$ in $S$. Since $f_{n+1}=2n+1$, this immediately implies that $n+1$ must be placed in positions $1$ and $2n+2$ in $S$, which in turn implies that positions $2,\ldots, 2n+1$ of~$S$ are covered by two visits of each $i\in [n]$. Thus, we can simply construct $S'$ by copying positions $2,\ldots, 2n+1$ of~$S$, with its feasibility following immediately from Def.~\ref{def:kVCov} and~\ref{def:VarPWCov}.
\end{proof}

\section{Positive results and tractable special cases for \twoVCov}

\subsection{A linear-time algorithm for two distinct frequencies}

In this subsection we limit \twoVCov\ to instances with up to two distinct frequencies. We will show a greedy linear-time algorithm for this case, loosely based on the observation that follows. Recall that we are allowed to use Def.~\ref{def:twoVCov} for \twoVCov, due to Lemma~\ref{lem:disconnection_property}.

\begin{definition}[Dominated triplet]\label{def:dominated_triplet}
    Let $F=\{f_1,\ldots,f_n\}$ be a \twoVCov\ instance and let $T_1=(f,b,c)$ and $T_2=(f',b',c')$ be two triplets in $F\times [n]\times [n+1,\ 2n]$ with $f+b \leq c$ and $f'+b'\leq c'$. Then, we say that $T_1$ \emph{dominates} $T_2$ if 
    $f\geq f'$, $b\geq b'$ and $c\leq c'$.\footnote{Note that we do not demand e.g. $f$ and $f'$ to be distinct. Trivially, any feasible triplet dominates itself.}
\end{definition}

Intuitively, in any feasible \twoVCov\ triplet $(f,b,c)$ we can swap $f$ or $b$ for a smaller number or $c$ for a larger number, while maintaining the triplet's feasibility. Hence, there is no reason to pick some triplet over one that dominates it (if both are available): the latter achieves feasibility while using elements that are, in a sense, harder to match than the elements of the former. We formalize this idea as follows.

\begin{observation}\label{obs:dominated_triplet}
    Suppose we have (irreversibly) matched some triplets in $F\times [n]\times [n+1,\ 2n]$ in an attempt to construct a \twoVCov\ solution for $F$, and let $T_1$ and $T_2$ be two feasible triplets consisting of unmatched elements in $F$, $[n]$ and $[n+1,\ 2n]$, with $T_1$ dominating $T_2$. Then, if there is a solution for $F$ using $T_2$ (along with all previously matched triplets), there is also a solution for $F$ using $T_1$ (along with all previously matched triplets). 
\end{observation}

The correctness of Observation~\ref{obs:dominated_triplet} formally stems from the fact that a solution using the elements remaining unmatched by picking $T_2$ can be transformed into a solution using the elements remaining unmatched by picking $T_1$, with simple substitution.

\begin{theorem}\label{thrm:two_distinct}
    There is an algorithm running in $\bigO(n)$ time that solves \twoVCov\ instances with up to two distinct frequencies.
\end{theorem}

\begin{proof}
    Let $F$ be a multiset consisting of $n_1\geq 0$ copies of frequency $x\geq 1$ and $n_2\geq 0$ copies of frequency~$y>x$, with $n_1+n_2=n$. We define $B=[n]$ and $C=[n+1,\ 2n]$.
    We can construct a solution for $F$ as follows.
    \begin{enumerate}
        \item If there are no unmatched copies of~$x$ left in~$F$, then it clearly suffices to match all copies of~$y$ with the elements in $B$ and $C$ in sorted order, and check whether this creates a feasible solution. Symmetrically, the same holds if no unmatched copies of $y$ are left in~$F$. For all subsequent cases, we assume that there is at least one unmatched copy of both $x$ and $y$ in $F$.
        \item If $y+\max(B) \leq \max(C)$, then $(y,\max(B),\max(C))$ is a feasible triplet and thus dominates all triplets that contain $\max(C)$. Since any solution must have a triplet that contains $\max(C)$, we can greedily pick the triplet $(y,\max(B),\max(C))$ by Observation~\ref{obs:dominated_triplet}.
        \item Otherwise, we have $y+\max(B) > \max(C)$. It follows that no triplet containing $y$ and $\max(B)$ is feasible. Thus, $\max(B)$ must be matched with $x$. Hence:
        \begin{itemize}
            \item If there is no feasible triplet containing $x$ and $\max(B)$, then no schedule exists.
            \item Otherwise, the feasible triplet $(x,\max(B),c)$ with minimum possible $c$ dominates all other triplets of this form. By Observation~\ref{obs:dominated_triplet}, we can greedily pick this triplet.
        \end{itemize}
        \item We remove the elements of the newly picked triplet from $F,B,C$ and repeat the process.
    \end{enumerate}

    The correctness of the described algorithm follows from the arguments above. The algorithm runs in time $\bigO(n)$, since in each of the at most $n$ repetitions we either pick a triplet and remove its elements from the input, or we conclude that no schedule exists. We remark that finding the minimum possible $c$ in step (3) requires $\bigO(n)$ time for all repetitions combined: in each repetition the sum $x+\max(B)$ can only decrease, therefore it suffices to store a pointer in the set $C$ showing the last-picked $c\in C$.
\end{proof}

\subsection{Parameterizing by the number of distinct frequencies}

In this subsection we parameterize \twoVCov\ by the number of distinct frequencies, also known as \emph{number of numbers} for numerical problems (cf.~\cite{Number_of_numbers,ICALP_kVisits}). We will prove that \twoVCov\ admits a randomized polynomial-time ($\mathrm{RP}$) algorithm when the number of numbers is constant, serving as a generalization of our algorithm in the previous subsection.

Our main tool here is the \EWPM\ (\EWPMshort) problem, which is known to admit a randomized polynomial-time algorithm for polynomially bounded weights through a reduction~\cite{Maalouly_stacs_2022,EWPM_to_EM_reduction} to \EM~\cite{Vazirani_EM,Papadimitriou_EM}. The derandomization of this algorithm is a long-standing open question in theoretical computer science~\cite{Maalouly_stacs_2022,Maalouly_esa_2025,Gurjar_EM_Complete}. Note that \EWPMshort\ is NP-complete when exponential weights are allowed~\cite{EWPM_to_EM_reduction}.

\begin{definition}[\EWPMshort]\label{def:EWPM}
    Given an edge-weighted (multi)graph $G=(V,E,w)$ and an integer $W$, the \EWPM\ $(\EWPMshort)$ problem asks whether there exists a perfect matching~$M$ in~$G$ with $\sum_{e \in M} w(e) =W$.
\end{definition}

\begin{remark}
    \EWPMshort\ is usually defined in \emph{simple} graphs in the literature. Regardless, there exists a polynomial-time reduction from \EWPMshort\ in multigraphs to \EWPMshort\ in simple graphs by Kanellopoulos et al.~\cite{ICALP_kVisits}. Thus, the aforementioned results for \EWPMshort\ generalize to multigraphs. For this reason, we state Def.~\ref{def:EWPM} directly for multigraphs.
\end{remark}

We will now reduce \twoVCov\ to \EWPMshort. We denote the number of numbers of a \twoVCov\ instance by $p$. This reduction uses the same core ideas as the one used in~\cite{ICALP_kVisits} for \twoVPack.

\begin{theorem}\label{thrm:2VCov_to_EWPM}
    \twoVCov\ reduces in time $\bigO(n^2p)$ to \EWPMshort\ with weights bounded by $\bigO(n^{p-1})$, where $p$ is the number of distinct frequencies in the input.
\end{theorem}

\begin{proof}
    Let $F$ be a \twoVCov\ instance consisting of $p$ distinct numbers $f_1,\ldots,f_p$ with multiplicities $n_1,\ldots,n_p$ respectively, where $n_1+\ldots+n_p=n$. By Lemma~\ref{lem:disconnection_property}, we can use Def.~\ref{def:twoVCov} for \twoVCov, hence the problem is reduced to finding a perfect matching between the sets~$F$, $B=[n]$ and $C=[n+1,\ 2n]$ such that for every triplet $(f,b,c)\in F\times B\times C$ in the matching it holds that $f+b \leq c$.

    We construct a bipartite multigraph $G$ as follows. We use $B=[n]$ and $C=[n+1,\ 2n]$ as the two vertex sets of $G$ (for simplicity, we denote vertices with numbers). For each pair $(b,c)\in B\times C$, we connect $b$ and $c$ with one edge for every distinct $f_i\in F$ ($i\in [p]$) that satisfies $f_i +b \leq c$; the weight of the respective edge is set to $w_i=(n+1)^{i-1}$ ($i\in [p]$). Note that parallel edges may exist, if more than one (distinct) $f_i$ satisfies $f_i +b \leq c$.

    We prove that $F$ admits a \twoVCov\ schedule if and only if $G$ with target weight $W=\sum_{i\in [p]} n_i w_i$ is a yes-instance of \EWPMshort.

    ($\Rightarrow$) Suppose $F$ admits a \twoVCov\ schedule, i.e., there exists a (perfect) matching $M\subseteq F\times B\times C$ such that $\forall (f,b,c)\in M$ it holds that $f+b\leq c$. By construction, for each triplet $(f_i,b,c)\in M$ there exists an edge between $b$ and $c$ with weight $w_i=(n+1)^{i-1}$ in $G$. Since $M$ is a (perfect) matching, each $b\in B$ and $c\in C$ appears in exactly one triplet in $M$. Thus, the set of the edges in $G$ that correspond to the $n$ triplets of $M$ is a perfect matching in $G$. The weight of this matching is $W=\sum_{i\in [p]} n_i w_i$, since each $f_i$ ($i\in [p]$) appears in exactly $n_i$ triplets of $M$.

    ($\Leftarrow$) Suppose $G$ has a perfect matching $M$ of total weight $W=\sum_{i\in [p]} n_i w_i$. Observe that there is a unique combination of $n$ weights in $G$ (of the form $w_i=(n+1)^{i-1}$, $i\in [p]$) with sum equal to~$W$; this is due to the uniqueness of base $(n+1)$ representation. Thus, the only way for $M$ to have total weight $W$ is to use $n_i$ edges of weight $w_i=(n+1)^{i-1}$, for all $i\in [p]$. It now follows that picking the triplets $(f,b,c)$ corresponding to the edges in $M$ induces a solution for $F$, by reversing the arguments presented in the forward direction of this proof.

    This concludes the correctness of the reduction. The time required to construct $G$ is $\bigO(n^2p)$.
\end{proof}

We obtain the following corollary from Theorem~\ref{thrm:2VCov_to_EWPM}, combined with a reduction from \EWPMshort\ to \EM~\cite{Maalouly_stacs_2022,EWPM_to_EM_reduction} and the randomized polynomial-time algorithm of Mulmuley, Vazirani and Vazirani~\cite{Vazirani_EM} for \EM. The corollary only holds for constant number of numbers, since the weights of the resulting \EWPMshort\ instance (which are $\bigO(n^{p-1})$ by Theorem~\ref{thrm:2VCov_to_EWPM}) must be polynomially bounded for the reduction to \EM\ to run in polynomial time~\cite{Maalouly_stacs_2022,EWPM_to_EM_reduction}.

\begin{corollary}\label{cor:RP}
    \twoVCov\ is in $\mathrm{RP}$ when the number of numbers is constant.
\end{corollary}

Since it is rare for a problem to be in $\mathrm{RP}$ and not in $\mathrm{P}$, we state the following conjecture.

\begin{conjecture}\label{conj:FPT_numbers}
    \twoVCov\ admits a (deterministic) XP algorithm parameterized by the number of numbers.
\end{conjecture}

Note that such an algorithm might also be FPT, since various numerical matching variants are known to be FPT by the number of numbers~\cite{Number_of_numbers}. The standard method of showing this through a reduction to \ILP~\cite{Number_of_numbers}, however, does not seem applicable to \twoVCov\ because two of three sets in Def.~\ref{def:twoVCov} consist of distinct numbers. Hence, we only conjecture an XP algorithm, due to the $\bigO(n^{p-1})$ weights introduced by Theorem~\ref{thrm:2VCov_to_EWPM}.

\subsection{An improved brute-force algorithm}

A naive brute-force algorithm for \twoVCov\ would run in time $\bigO((2n)!)$. However, from Lemma~\ref{lem:disconnection_property} we can immediately obtain the following improved running time. 

\begin{corollary}
    There is a brute-force algorithm running in time $\bigO(n!)$ for \twoVCov.
\end{corollary}

\section{On the density thresholds of \kVCov}

In this section, we study density bounds for \kVCov. The \emph{density} of an instance $F=\{f_1,\ldots,f_n\}$ is defined as $\mathrm{Dens}(F)=\sum_{i=1}^n 1/f_i$.
We use the following terminology (cf.~\cite{ICALP_kVisits}):
\begin{itemize}
    \item The maximum density below which no instance of a problem is schedulable is the \emph{lower density threshold} of that problem.
    \item The minimum density above which all instances of a problem are schedulable is the \emph{upper density threshold} of that problem.
\end{itemize}

Intuitively, increasing the density of an instance makes it easier to schedule in the \emph{covering} setting, but harder in the \emph{packing} setting. As such, the definitions of lower and upper density thresholds presented here are, in a sense, reversed, when compared to the analogous definitions in the packing setting (e.g.,~\cite{ICALP_kVisits}).

We remark that for \PWCov\ the lower density threshold is $1$ (trivial, cf.~\cite{Kawamura_pinwheel_cover}) and the upper density threshold is $\sum_{i=0}^\infty 1/(2^i+1)\approx 1.264$~\cite{Kawamura_Kobayashi_density_covering,Patrolling_SODA_2026}. Both of these values are tight. Here, our goal is to investigate analogous thresholds for \kVCov.

\subsection{Lower density threshold}

We prove that the lower density threshold of $1$ known for \PWCov~\cite{Kawamura_pinwheel_cover} also holds for \kVCov, for all $k\geq 2$. In contrast to the infinite version, this is not immediately obvious for the finite one. We will bound the sum of frequencies for schedulable instances and utilize the arithmetic-harmonic mean inequality to obtain a bound for the density.

\begin{theorem}\label{thrm:lower_density_threshold}
    The lower density threshold of \kVCov\ is $1$ for all $k\geq 2$.
\end{theorem}

\begin{proof}
    Let $F=\{f_1,\ldots,f_n\}$ be a schedulable \kVCov\ instance ($k\geq 2$).
    For schedule~$S$ of~$F$, define $\mathrm{dist}(i)$ as the distance between the first and last occurrence of $i\in [n]$ in $S$. We aim to bound the quantity $\sum_{i\in [n]} \mathrm{dist}(i)$. Observe that this quantity is trivially bounded by $kn^2$, however, we require a slightly better bound for our proof. Let $\mathrm{first}(i)$ and $\mathrm{last}(i)$ be the positions of the first and last occurrence of $i$ in $S$. We have
    \[\sum_{i\in [n]} \mathrm{first}(i) \geq 1+2+\ldots+n,\quad \sum_{i\in [n]} \mathrm{last}(i) \leq kn+(kn-1)+\ldots+(kn-(n-1)).\]
    We obtain
    \begin{equation}\label{eq:distance_bound}
        \sum_{i\in [n]} \mathrm{dist}(i) \leq \sum_{i\in [n]} (\mathrm{last}(i)- \mathrm{first}(i))\leq (k-1)n^2.
    \end{equation}
    Now, observe that $\mathrm{dist}(i)\geq (k-1)f_i$, by definition. Combining this with \eqref{eq:distance_bound}, we obtain
    \begin{equation}\label{eq:frequency_sum_bound}
        \sum_{i\in [n]} f_i \leq n^2.
    \end{equation}
    By the arithmetic-harmonic mean inequality it holds that
    \begin{equation}\label{eq:AM_HM}
        \sum_{i\in [n]} f_i \geq \frac{n^2}{\sum_{i\in [n]} 1/f_i} = \frac{n^2}{\mathrm{Dens}(F)}.
    \end{equation}
    From \eqref{eq:frequency_sum_bound} and \eqref{eq:AM_HM}, we obtain $\mathrm{Dens}(F)\geq 1$. Thus, if $\mathrm{Dens}(F)< 1$, then $F$ is unschedulable.     

    The instance $\{1\}$ trivially admits a \kVCov\ schedule for every $k\in \mathbb{N}$, therefore the threshold of $1$ that we proved above is tight for all $k\geq 2$.
\end{proof}

\subsection{Upper density threshold}

Observe that every instance containing a frequency $f \geq nk$ cannot admit a \kVCov\ schedule, for every $k\geq 2$. As such, $\{1,1,\ldots,1,nk\}$ (with $n-1$ copies of $1$) is an unschedulable instance of \kVCov\ with arbitrarily large density (as $n$ increases). Thus, we obtain the following.

\begin{theorem}\label{thrm:upper_density_threshold}
    No upper density threshold exists for \kVCov\ for any $k\geq 2$.
\end{theorem}

A natural question now is whether an upper density threshold would exist if we prevented the trivial and arguably meaningless edge case of some frequency being large enough to make the instance unschedulable as above. However, we can prove that this is not the case. For example, for $k=2$, the instance $F=\{1,3,5,\ldots,2n-1\}$ admits the schedule $2n-1,\ldots,3,1,1,3,\ldots,2n-1$ and has arbitrarily large density (due to the harmonic series diverging). Observe that increasing any frequency in $F$ by $1$ makes the instance unschedulable, giving an alternative proof for Theorem~\ref{thrm:upper_density_threshold} (for $k=2$). It is straightforward to extend a similar argument to all $k\geq 2$.

\section{Conclusion}

As in the packing setting, the most important question arising from this work is whether some finite variant of \PWCov\ can be used to transfer hardness results to the infinite version. We already showed that this is the case with \VarPWCov\ being strongly NP-hard. Note that finite versions may also prove useful for PSPACE-hardness proofs, since a standard method to prove PSPACE-hardness relies on modifying the NP-hardness proof of a finite version (cf.~\cite{Los_Alamos_periodic_PSPACE,PSPACE_dynamic,Papadimitriou_book}), although it remains unclear whether this can be applied to perpetual scheduling variants.

Another natural open question is whether the disconnection property (Lemma~\ref{lem:disconnection_property}) we showed for \twoVCov\ generalizes to \kVCov\ ($k\geq 3$). If this is the case, then it may also be possible to generalize the algorithms we showed for \twoVCov\ with constant amounts of distinct frequencies. We have no strong indication that the disconnection property cannot be generalized, other than the fact that our proof for Lemma~\ref{lem:disconnection_property} does not appear to apply for $k\geq 3$. Interestingly, the analogous disconnection property of \twoVPack~\cite{kVisits} has been disproved for \threeVPack~\cite{ICALP_kVisits}.

We remark that, in contrast to \twoVPack~\cite{kVisits}, \twoVCov\ does not seem to be tractable when all input numbers are distinct. Our reduction in this paper does not immediately prove this, since \RNTDMshort\ is not known to be NP-hard with distinct numbers (to the best of our knowledge). Regardless, it would be surprising if \RNTDMshort\ is tractable for distinct numbers (cf.~\cite{NMTS_distinct}), hence we state the following conjecture. Note that the known NP-hardness proof for \RNTDMshort~\cite{Flow_shop_Yu} involves a complicated reduction that pads \tP\ with duplicate numbers, and modifying it to forgo these duplicates seems challenging.

\begin{conjecture}
    \twoVCov\ is strongly NP-complete even when $F$ is a simple set, contrasting \twoVPack.
\end{conjecture}

Lastly, we remark that, according to Theorem~\ref{thrm:upper_density_threshold}, \kVCov\ with $k\to \infty$ does not approach the known upper density threshold of $1.264\ldots$ for \PWCov~\cite{Kawamura_Kobayashi_density_covering,Patrolling_SODA_2026}. On the contrary, \kVPack\ is known to approach the $5/6$-threshold of \PWPack\ for $k\to \infty$~\cite{ICALP_kVisits}. We leave as a direction for future work the establishment of an alternative finite covering version that respects the upper density threshold of \PWCov.

\subsubsection*{Acknowledgments} 

The author is grateful to Aris Pagourtzis, Christos Pergaminelis, Leszek Gasieniec, Karteek Sreenivasaiah and Prudence Wong for valuable discussions regarding Pinwheel Scheduling and related problems. This work has been partially supported by project MIS 5154714 of the National Recovery and Resilience Plan Greece 2.0 funded by the European Union under the NextGenerationEU Program. The author acknowledges the financial support of the Archimedes Research Unit (Athena RC) during his PhD.

No artificial intelligence was used at any point during the production of this paper. All of the presented results were conceived exclusively by the author and the entirety of the paper was written exclusively by the author.

\bibliographystyle{splncs04}
\bibliography{bibliography}

@INPROCEEDINGS{Holte_Pinwheel,
  author={Holte, R. and Mok, A. and Rosier, L. and Tulchinsky, I. and Varvel, D.},
  booktitle={[1989] Proceedings of the Twenty-Second Annual Hawaii International Conference on System Sciences. Volume II: Software Track}, 
  title={The pinwheel: a real-time scheduling problem}, 
  year={1989},
  volume={2},
  number={},
  pages={693-702 vol.2},
  doi={10.1109/HICSS.1989.48075}}

@inproceedings{Gasieniec_towards_5/6,
  author       = {Leszek Gasieniec and
                  Benjamin Smith and
                  Sebastian Wild},
  editor       = {Cynthia A. Phillips and
                  Bettina Speckmann},
  title        = {Towards the 5/6-Density Conjecture of Pinwheel Scheduling},
  booktitle    = {Proceedings of the Symposium on Algorithm Engineering and Experiments,
                  {ALENEX} 2022, Alexandria, VA, USA, January 9-10, 2022},
  pages        = {91--103},
  publisher    = {{SIAM}},
  year         = {2022},
  url          = {https://doi.org/10.1137/1.9781611977042.8},
  doi          = {10.1137/1.9781611977042.8},
  bibsource    = {dblp computer science bibliography, https://dblp.org}
}

@inproceedings{Kawamura_5/6_stoc,
  author       = {Akitoshi Kawamura},
  editor       = {Bojan Mohar and
                  Igor Shinkar and
                  Ryan O'Donnell},
  title        = {Proof of the Density Threshold Conjecture for Pinwheel Scheduling},
  booktitle    = {Proceedings of the 56th Annual {ACM} Symposium on Theory of Computing,
                  {STOC} 2024, Vancouver, BC, Canada, June 24-28, 2024},
  pages        = {1816--1819},
  publisher    = {{ACM}},
  year         = {2024},
  url          = {https://doi.org/10.1145/3618260.3649757},
  doi          = {10.1145/3618260.3649757},
  bibsource    = {dblp computer science bibliography, https://dblp.org}
}

@inproceedings{Kawamura_pinwheel_cover,
  author       = {Akitoshi Kawamura and
                  Yusuke Kobayashi and
                  Yosuke Kusano},
  editor       = {Irene Finocchi and
                  Loukas Georgiadis},
  title        = {Pinwheel Covering},
  booktitle    = {Algorithms and Complexity - 14th International Conference, {CIAC}
                  2025, Rome, Italy, June 10-12, 2025, Proceedings, Part {II}},
  series       = {Lecture Notes in Computer Science},
  volume       = {15680},
  pages        = {185--199},
  publisher    = {Springer},
  year         = {2025},
  url          = {https://doi.org/10.1007/978-3-031-92935-9\_12},
  doi          = {10.1007/978-3-031-92935-9\_12},
  bibsource    = {dblp computer science bibliography, https://dblp.org}
}

@article{Feinberg_Generalized_Pinwheel,
  author       = {Eugene A. Feinberg and
                  Michael T. Curry},
  title        = {Generalized Pinwheel Problem},
  journal      = {Math. Methods Oper. Res.},
  volume       = {62},
  number       = {1},
  pages        = {99--122},
  year         = {2005},
  url          = {https://doi.org/10.1007/s00186-005-0443-4},
  doi          = {10.1007/S00186-005-0443-4},
  bibsource    = {dblp computer science bibliography, https://dblp.org}
}

@inproceedings{Bamboo_approx_1,
  author       = {Felix H{\"{o}}hne and
                  Rob van Stee},
  editor       = {Nicole Megow and
                  Adam D. Smith},
  title        = {A 10/7-Approximation for Discrete Bamboo Garden Trimming and Continuous
                  Trimming on Star Graphs},
  booktitle    = {Approximation, Randomization, and Combinatorial Optimization. Algorithms
                  and Techniques, {APPROX/RANDOM} 2023, September 11-13, 2023, Atlanta,
                  Georgia, {USA}},
  series       = {LIPIcs},
  volume       = {275},
  pages        = {16:1--16:19},
  publisher    = {Schloss Dagstuhl - Leibniz-Zentrum f{\"{u}}r Informatik},
  year         = {2023},
  url          = {https://doi.org/10.4230/LIPIcs.APPROX/RANDOM.2023.16},
  doi          = {10.4230/LIPICS.APPROX/RANDOM.2023.16},
  bibsource    = {dblp computer science bibliography, https://dblp.org}
}

@article{Bamboo_approx_2,
  author       = {Martijn van Ee},
  title        = {A 12/7-approximation algorithm for the discrete Bamboo Garden Trimming
                  problem},
  journal      = {Oper. Res. Lett.},
  volume       = {49},
  number       = {5},
  pages        = {645--649},
  year         = {2021},
  url          = {https://doi.org/10.1016/j.orl.2021.07.001},
  doi          = {10.1016/J.ORL.2021.07.001},
  bibsource    = {dblp computer science bibliography, https://dblp.org}
}

@article{Jacobs_Window_Scheduling_Complexity,
  author       = {Tobias Jacobs and
                  Salvatore Longo},
  title        = {A New Perspective on the Windows Scheduling Problem},
  journal      = {CoRR},
  volume       = {abs/1410.7237},
  year         = {2014},
  url          = {http://arxiv.org/abs/1410.7237},
  eprinttype   = {arXiv},
  eprint       = {1410.7237},
  bibsource    = {dblp computer science bibliography, https://dblp.org}
}

@article{Bamboo_second,
  author       = {Mattia D'Emidio and
                  Gabriele Di Stefano and
                  Alfredo Navarra},
  title        = {Bamboo Garden Trimming Problem: Priority Schedulings},
  journal      = {Algorithms},
  volume       = {12},
  number       = {4},
  pages        = {74},
  year         = {2019},
  url          = {https://doi.org/10.3390/a12040074},
  doi          = {10.3390/A12040074},
  bibsource    = {dblp computer science bibliography, https://dblp.org}
}

@inproceedings{Bamboo_first,
  author       = {Leszek Gasieniec and
                  Ralf Klasing and
                  Christos Levcopoulos and
                  Andrzej Lingas and
                  Jie Min and
                  Tomasz Radzik},
  editor       = {Bernhard Steffen and
                  Christel Baier and
                  Mark van den Brand and
                  Johann Eder and
                  Mike Hinchey and
                  Tiziana Margaria},
  title        = {Bamboo Garden Trimming Problem (Perpetual Maintenance of Machines
                  with Different Attendance Urgency Factors)},
  booktitle    = {{SOFSEM} 2017: Theory and Practice of Computer Science - 43rd International
                  Conference on Current Trends in Theory and Practice of Computer Science,
                  Limerick, Ireland, January 16-20, 2017, Proceedings},
  series       = {Lecture Notes in Computer Science},
  volume       = {10139},
  pages        = {229--240},
  publisher    = {Springer},
  year         = {2017},
  url          = {https://doi.org/10.1007/978-3-319-51963-0\_18},
  doi          = {10.1007/978-3-319-51963-0\_18},
  bibsource    = {dblp computer science bibliography, https://dblp.org}
}

@article{Flow_shop_Yu,
  author       = {Wenci Yu and
                  Han Hoogeveen and
                  Jan Karel Lenstra},
  title        = {Minimizing Makespan in a Two-Machine Flow Shop with Delays and Unit-Time
                  Operations is {NP}-Hard},
  journal      = {J. Sched.},
  volume       = {7},
  number       = {5},
  pages        = {333--348},
  year         = {2004},
  url          = {https://doi.org/10.1023/B:JOSH.0000036858.59787.c2},
  doi          = {10.1023/B:JOSH.0000036858.59787.C2},
  bibsource    = {dblp computer science bibliography, https://dblp.org}
}

@article{Number_of_numbers,
  author       = {Michael R. Fellows and
                  Serge Gaspers and
                  Frances A. Rosamond},
  title        = {Parameterizing by the Number of Numbers},
  journal      = {Theory Comput. Syst.},
  volume       = {50},
  number       = {4},
  pages        = {675--693},
  year         = {2012},
  url          = {https://doi.org/10.1007/s00224-011-9367-y},
  doi          = {10.1007/S00224-011-9367-Y},
  bibsource    = {dblp computer science bibliography, https://dblp.org}
}

@article{Chan_conjecture,
  author       = {Mee Yee Chan and
                  Francis Y. L. Chin},
  title        = {Schedulers for Larger Classes of Pinwheel Instances},
  journal      = {Algorithmica},
  volume       = {9},
  number       = {5},
  pages        = {425--462},
  year         = {1993},
  url          = {https://doi.org/10.1007/BF01187034},
  doi          = {10.1007/BF01187034},
  bibsource    = {dblp computer science bibliography, https://dblp.org}
}

@article{Chan_0.7,
  author       = {Mee Yee Chan and
                  Francis Y. L. Chin},
  title        = {General Schedulers for the Pinwheel Problem Based on Double-Integer
                  Reduction},
  journal      = {{IEEE} Trans. Computers},
  volume       = {41},
  number       = {6},
  pages        = {755--768},
  year         = {1992},
  url          = {https://doi.org/10.1109/12.144627},
  doi          = {10.1109/12.144627},
  bibsource    = {dblp computer science bibliography, https://dblp.org}
}

@article{Bar-Noy_0.6,
  author       = {Amotz Bar{-}Noy and
                  Richard E. Ladner and
                  Tami Tamir},
  title        = {Windows scheduling as a restricted version of bin packing},
  journal      = {{ACM} Trans. Algorithms},
  volume       = {3},
  number       = {3},
  pages        = {28},
  year         = {2007},
  url          = {https://doi.org/10.1145/1273340.1273344},
  doi          = {10.1145/1273340.1273344},
  bibsource    = {dblp computer science bibliography, https://dblp.org}
}

@article{Fishburn_density,
  author       = {Peter C. Fishburn and
                  J. C. Lagarias},
  title        = {Pinwheel Scheduling: Achievable Densities},
  journal      = {Algorithmica},
  volume       = {34},
  number       = {1},
  pages        = {14--38},
  year         = {2002},
  url          = {https://doi.org/10.1007/s00453-002-0938-9},
  doi          = {10.1007/S00453-002-0938-9},
  bibsource    = {dblp computer science bibliography, https://dblp.org}
}

@inproceedings{kVisits,
  author       = {Sotiris Kanellopoulos and
                  Christos Pergaminelis and
                  Maria Kokkou and
                  Euripides Markou and
                  Aris Pagourtzis},
  editor       = {Kasper Green Larsen and
                  Barna Saha},
  title        = {Finite Pinwheel Scheduling: the k-Visits Problem},
  booktitle    = {Proceedings of the 2026 Annual {ACM-SIAM} Symposium on Discrete Algorithms,
                  {SODA} 2026, Vancouver, BC, Canada, January 11-14, 2026},
  pages        = {355--371},
  publisher    = {{SIAM}},
  year         = {2026},
  url          = {https://doi.org/10.1137/1.9781611978971.16},
  doi          = {10.1137/1.9781611978971.16},
  bibsource    = {dblp computer science bibliography, https://dblp.org}
}

@article{NMTS_distinct,
  author      = {Heather Hulett and
                  Todd G. Will and
                  Gerhard J. Woeginger},
  title        = {Multigraph realizations of degree sequences: Maximization is easy,
                  minimization is hard},
  journal      = {Oper. Res. Lett.},
  volume       = {36},
  number       = {5},
  pages        = {594--596},
  year         = {2008},
  url          = {https://doi.org/10.1016/j.orl.2008.05.004},
  doi          = {10.1016/J.ORL.2008.05.004},
  bibsource    = {dblp computer science bibliography, https://dblp.org}
}

@article{Papadimitriou_EM,
  author       = {Christos H. Papadimitriou and
                  Mihalis Yannakakis},
  title        = {The complexity of restricted spanning tree problems},
  journal      = {J. {ACM}},
  volume       = {29},
  number       = {2},
  pages        = {285--309},
  year         = {1982},
  url          = {https://doi.org/10.1145/322307.322309},
  doi          = {10.1145/322307.322309},
  bibsource    = {dblp computer science bibliography, https://dblp.org}
}

@inproceedings{Vazirani_EM,
  author       = {Ketan Mulmuley and
                  Umesh V. Vazirani and
                  Vijay V. Vazirani},
  editor       = {Alfred V. Aho},
  title        = {Matching Is as Easy as Matrix Inversion},
  booktitle    = {Proceedings of the 19th Annual {ACM} Symposium on Theory of Computing,
                  1987, New York, New York, {USA}},
  pages        = {345--354},
  publisher    = {{ACM}},
  year         = {1987},
  url          = {https://doi.org/10.1145/28395.383347},
  doi          = {10.1145/28395.383347},
  bibsource    = {dblp computer science bibliography, https://dblp.org}
}

@article{Gurjar_EM_Complete,
  author       = {Rohit Gurjar and
                  Arpita Korwar and
                  Jochen Messner and
                  Thomas Thierauf},
  title        = {Exact Perfect Matching in Complete Graphs},
  journal      = {{ACM} Trans. Comput. Theory},
  volume       = {9},
  number       = {2},
  pages        = {8:1--8:20},
  year         = {2017},
  url          = {https://doi.org/10.1145/3041402},
  doi          = {10.1145/3041402},
  bibsource    = {dblp computer science bibliography, https://dblp.org}
}

@InProceedings{Maalouly_stacs_2022,
  author =	{El Maalouly, Nicolas},
  title =	{{Exact Matching: Algorithms and Related Problems}},
  booktitle =	{40th International Symposium on Theoretical Aspects of Computer Science (STACS 2023)},
  pages =	{29:1--29:17},
  series =	{Leibniz International Proceedings in Informatics (LIPIcs)},
  ISBN =	{978-3-95977-266-2},
  ISSN =	{1868-8969},
  year =	{2023},
  volume =	{254},
  editor =	{Berenbrink, Petra and Bouyer, Patricia and Dawar, Anuj and Kant\'{e}, Mamadou Moustapha},
  publisher =	{Schloss Dagstuhl -- Leibniz-Zentrum f{\"u}r Informatik},
  address =	{Dagstuhl, Germany},
  URL =		{https://drops.dagstuhl.de/entities/document/10.4230/LIPIcs.STACS.2023.29},
  URN =		{urn:nbn:de:0030-drops-176811},
  doi =		{10.4230/LIPIcs.STACS.2023.29}
}

@InProceedings{Maalouly_esa_2025,
  author =	{El Maalouly, Nicolas and Haslebacher, Sebastian and Taubner, Adrian and Wulf, Lasse},
  title =	{{On Finding l-Th Smallest Perfect Matchings}},
  booktitle =	{33rd Annual European Symposium on Algorithms (ESA 2025)},
  pages =	{19:1--19:15},
  series =	{Leibniz International Proceedings in Informatics (LIPIcs)},
  ISBN =	{978-3-95977-395-9},
  ISSN =	{1868-8969},
  year =	{2025},
  volume =	{351},
  editor =	{Benoit, Anne and Kaplan, Haim and Wild, Sebastian and Herman, Grzegorz},
  publisher =	{Schloss Dagstuhl -- Leibniz-Zentrum f{\"u}r Informatik},
  address =	{Dagstuhl, Germany},
  URL =		{https://drops.dagstuhl.de/entities/document/10.4230/LIPIcs.ESA.2025.19},
  URN =		{urn:nbn:de:0030-drops-244875},
  doi =		{10.4230/LIPIcs.ESA.2025.19}
}

@inproceedings{Pinwheel_ISAAC_2025,
  author       = {Yusuke Kobayashi and
                  Bingkai Lin},
  editor       = {Ho{-}Lin Chen and
                  Wing{-}Kai Hon and
                  Meng{-}Tsung Tsai},
  title        = {Hardness and Fixed Parameter Tractability for Pinwheel Scheduling
                  Problems},
  booktitle    = {36th International Symposium on Algorithms and Computation, {ISAAC}
                  2025, Tainan, Taiwan, December 7-10, 2025},
  series       = {LIPIcs},
  volume       = {359},
  pages        = {47:1--47:15},
  publisher    = {Schloss Dagstuhl - Leibniz-Zentrum f{\"{u}}r Informatik},
  year         = {2025},
  url          = {https://doi.org/10.4230/LIPIcs.ISAAC.2025.47},
  doi          = {10.4230/LIPICS.ISAAC.2025.47},
  bibsource    = {dblp computer science bibliography, https://dblp.org}
}

@inproceedings{Patrolling_SODA_2026,
  title={An Optimal Density Bound for Discretized Point Patrolling},
  author={Mishra, Ahan},
  booktitle={Proceedings of the 2026 Annual ACM-SIAM Symposium on Discrete Algorithms (SODA)},
  pages={2846--2875},
  year={2026},
  organization={SIAM},
  URL = {https://epubs.siam.org/doi/abs/10.1137/1.9781611978971.105}
}

@inproceedings{PSPACE_UAV,
  author       = {Hsi{-}Ming Ho and
                  Jo{\"{e}}l Ouaknine},
  editor       = {Andrew M. Pitts},
  title        = {The Cyclic-Routing {UAV} Problem is PSPACE-Complete},
  booktitle    = {Foundations of Software Science and Computation Structures - 18th
                  International Conference, FoSSaCS 2015, Held as Part of the European
                  Joint Conferences on Theory and Practice of Software, {ETAPS} 2015,
                  London, UK, April 11-18, 2015. Proceedings},
  series       = {Lecture Notes in Computer Science},
  volume       = {9034},
  pages        = {328--342},
  publisher    = {Springer},
  year         = {2015},
  url          = {https://doi.org/10.1007/978-3-662-46678-0\_21},
  doi          = {10.1007/978-3-662-46678-0\_21},
  bibsource    = {dblp computer science bibliography, https://dblp.org}
}

@inproceedings{Los_Alamos_periodic_PSPACE,
  author       = {Madhav V. Marathe and
                  Harry B. Hunt III and
                  Daniel J. Rosenkrantz and
                  Richard Edwin Stearns},
  title        = {Theory of Periodically Specified Problems: Complexity and Approximability},
  booktitle    = {Proceedings of the 13th Annual {IEEE} Conference on Computational
                  Complexity, Buffalo, New York, USA, June 15-18, 1998},
  pages        = {106},
  publisher    = {{IEEE} Computer Society},
  year         = {1998},
  url          = {https://doi.org/10.1109/CCC.1998.694596},
  doi          = {10.1109/CCC.1998.694596},
  bibsource    = {dblp computer science bibliography, https://dblp.org}
}

@book{Papadimitriou_book,
  author       = {Christos H. Papadimitriou},
  title        = {Computational complexity},
  publisher    = {Addison-Wesley},
  year         = {1994},
  isbn         = {978-0-201-53082-7},
  bibsource    = {dblp computer science bibliography, https://dblp.org}
}

@article{EWPM_to_EM_reduction,
  author       = {Rohit Gurjar and
                  Arpita Korwar and
                  Jochen Messner and
                  Simon Straub and
                  Thomas Thierauf},
  title        = {Planarizing Gadgets for Perfect Matching Do Not Exist},
  journal      = {{ACM} Trans. Comput. Theory},
  volume       = {8},
  number       = {4},
  pages        = {14:1--14:15},
  year         = {2016},
  url          = {https://doi.org/10.1145/2934310},
  doi          = {10.1145/2934310},
  bibsource    = {dblp computer science bibliography, https://dblp.org}
}

@misc{kleinberg_mishra_NP_hardness,
      title={{NP}-Hardness and a {PTAS} for the Pinwheel Problem}, 
      author={Robert Kleinberg and Ahan Mishra},
      year={2026},
      eprint={2604.13974},
      archivePrefix={arXiv},
      primaryClass={cs.DS},
      url={https://arxiv.org/abs/2604.13974}, 
}

@inproceedings{ICALP_kVisits,
  author       = {Sotiris Kanellopoulos and
                  Giorgos Mitropoulos and
                  Christos Pergaminelis and
                  Thanos Tolias},
  editor       = {Sayan Bhattacharya and
                  Danupon Nanongkai and
                  Michael Benedikt and
                  Gabriele Puppis},
  title        = {Hardness, Tractability and Density Thresholds of Finite Pinwheel Scheduling
                  Variants},
  booktitle    = {53rd International Colloquium on Automata, Languages, and Programming,
                  {ICALP} 2026, Royal Holloway, University of London, Egham, United
                  Kingdom, July 7-10, 2026},
  series       = {LIPIcs},
  volume       = {374},
  pages        = {122:1--122:23},
  publisher    = {Schloss Dagstuhl - Leibniz-Zentrum f{\"{u}}r Informatik},
  year         = {2026},
  url          = {https://doi.org/10.4230/LIPIcs.ICALP.2026.122},
  doi          = {10.4230/LIPICS.ICALP.2026.122},
  bibsource    = {dblp computer science bibliography, https://dblp.org}
}

@misc{Kawamura_Kobayashi_density_covering,
      title={A Computer-Assisted Proof of the Optimal Density Bound for Pinwheel Covering}, 
      author={Akitoshi Kawamura and Yusuke Kobayashi},
      year={2025},
      eprint={2510.06533},
      archivePrefix={arXiv},
      primaryClass={cs.DM},
      url={https://arxiv.org/abs/2510.06533}, 
}

@article{Kawamura_Soejima_point_patrolling,
  author       = {Akitoshi Kawamura and
                  Makoto Soejima},
  title        = {Simple strategies versus optimal schedules in multi-agent patrolling},
  journal      = {Theor. Comput. Sci.},
  volume       = {839},
  pages        = {195--206},
  year         = {2020},
  url          = {https://doi.org/10.1016/j.tcs.2020.07.037},
  doi          = {10.1016/J.TCS.2020.07.037},
  bibsource    = {dblp computer science bibliography, https://dblp.org}
}

@inproceedings{SOFSEM_Kusano_pinwheel_durations,
  author       = {Yosuke Kusano},
  editor       = {Jakub Kozik and
                  Alexander Wolff},
  title        = {Limitations of Density-Based Heuristics and an Alternative Approach
                  for Pinwheel Scheduling with Durations},
  booktitle    = {{SOFSEM} 2026: Theory and Practice of Computer Science - 51st International
                  Conference on Current Trends in Theory and Practice of Computer Science,
                  {SOFSEM} 2026, Krak{\'{o}}w, Poland, February 9-13, 2026, Proceedings},
  series       = {Lecture Notes in Computer Science},
  pages        = {634--647},
  publisher    = {Springer},
  year         = {2026},
  url          = {https://doi.org/10.1007/978-3-032-17801-5\_46},
  doi          = {10.1007/978-3-032-17801-5\_46},
  bibsource    = {dblp computer science bibliography, https://dblp.org}
}

@inproceedings{SOFSEM_real_periods,
  author       = {Hiroshi Fujiwara and
                  Kota Miyagi and
                  Katsuhisa Ouchi},
  editor       = {Jakub Kozik and
                  Alexander Wolff},
  title        = {Pinwheel Scheduling with Real Periods},
  booktitle    = {{SOFSEM} 2026: Theory and Practice of Computer Science - 51st International
                  Conference on Current Trends in Theory and Practice of Computer Science,
                  {SOFSEM} 2026, Krak{\'{o}}w, Poland, February 9-13, 2026, Proceedings},
  series       = {Lecture Notes in Computer Science},
  pages        = {621--633},
  publisher    = {Springer},
  year         = {2026},
  url          = {https://doi.org/10.1007/978-3-032-17801-5\_45},
  doi          = {10.1007/978-3-032-17801-5\_45},
  bibsource    = {dblp computer science bibliography, https://dblp.org}
}

@InProceedings{Schewior_Combinatorial_Perpetual_Scheduling_ICALP,
  author =	{Mendoza-Cadena, Mirabel and Merino, Arturo and Nielsen, Mads Anker and Schewior, Kevin},
  title =	{{Combinatorial Perpetual Scheduling: Existence and Computation of Low-Height Schedules}},
  booktitle =	{53rd International Colloquium on Automata, Languages, and Programming (ICALP 2026)},
  pages =	{142:1--142:23},
  series =	{Leibniz International Proceedings in Informatics (LIPIcs)},
  ISBN =	{978-3-95977-428-4},
  ISSN =	{1868-8969},
  year =	{2026},
  volume =	{374},
  editor =	{Bhattacharya, Sayan and Nanongkai, Danupon and Benedikt, Michael and Puppis, Gabriele},
  publisher =	{Schloss Dagstuhl -- Leibniz-Zentrum f{\"u}r Informatik},
  address =	{Dagstuhl, Germany},
  URL =		{https://drops.dagstuhl.de/entities/document/10.4230/LIPIcs.ICALP.2026.142},
  URN =		{urn:nbn:de:0030-drops-265319},
  doi =		{10.4230/LIPIcs.ICALP.2026.142}
}

@inproceedings{PSPACE_dynamic,
  author       = {James B. Orlin},
  editor       = {Ronald L. Rivest and
                  Georgia I. Davida and
                  Walter A. Burkhard and
                  Richard J. Lipton},
  title        = {The Complexity of Dynamic Languages and Dynamic Optimization Problems},
  booktitle    = {Proceedings of the 13th Annual {ACM} Symposium on Theory of Computing,
                  May 11-13, 1981, Milwaukee, Wisconsin, {USA}},
  pages        = {218--227},
  publisher    = {{ACM}},
  year         = {1981},
  url          = {https://doi.org/10.1145/800076.802475},
  doi          = {10.1145/800076.802475},
  bibsource    = {dblp computer science bibliography, https://dblp.org}
}

@inproceedings{telephone_broadcasting_NMTS,
  author       = {Yudai Egami and
                  Tatsuya Gima and
                  Tesshu Hanaka and
                  Yasuaki Kobayashi and
                  Michael Lampis and
                  Valia Mitsou and
                  Edouard Nemery and
                  Yota Otachi and
                  Manolis Vasilakis and
                  Daniel Vaz},
  editor       = {Pawel Gawrychowski and
                  Filip Mazowiecki and
                  Michal Skrzypczak},
  title        = {Broadcasting Under Structural Restrictions},
  booktitle    = {50th International Symposium on Mathematical Foundations of Computer
                  Science, {MFCS} 2025, Warsaw, Poland, August 25-29, 2025},
  series       = {LIPIcs},
  volume       = {345},
  pages        = {42:1--42:18},
  publisher    = {Schloss Dagstuhl - Leibniz-Zentrum f{\"{u}}r Informatik},
  year         = {2025},
  url          = {https://doi.org/10.4230/LIPIcs.MFCS.2025.42},
  doi          = {10.4230/LIPICS.MFCS.2025.42},
  bibsource    = {dblp computer science bibliography, https://dblp.org}
}

@misc{temporal_path_cover,
      title={Temporal Path Covers: Dilworth Properties and Parameterized Complexity}, 
      author={Lapo Cioni and Sotiris Kanellopoulos and Edouard Nemery and Aris Pagourtzis and Christos Pergaminelis and Manolis Vasilakis},
      year={2026},
      eprint={2607.00118},
      archivePrefix={arXiv},
      primaryClass={cs.DS},
      url={https://arxiv.org/abs/2607.00118}, 
}


\newpage

\appendix

\end{document}